\documentclass{amsart}
\usepackage{amssymb,amsmath}

\newtheorem{definition}{Definition}[section]
\newtheorem{corollary}[definition]{Corollary}

\newtheorem{lemma}[definition]{Lemma}

\newtheorem{property}[definition]{Property}
\newtheorem{remark}[definition]{Remark}
\newtheorem{theorem}[definition]{Theorem}

\title[A relativized reducibility method]{Confluent terminating
extensional
lambda-calculi with surjective pairing and terminal type}
\author{YOHJI AKAMA}
\address{
Mathematical Institute, Tohoku University, Aoba, Sendai, JAPAN, 980-8578}

\def\A{\mathcal{A}}
\def\all#1#2{\Pi#1.\, #2}
\def\c{*}
\def\d#1#2#3{\mathrm{d}(#1,\ #2,\ #3)}

\def\etatop{\eta_{top}}
\def\FTV{\mathrm{FTV}}
\def\FV{\mathrm{FV}}
\def\gentop{{g}}

\def\ih#1{induction hypothesis (CR#1)}
\def\isomorphicT{\mathit{Iso}(\top)}
\def\isot{\in \isomorphicT}

\def\inj#1#2#3{\mathrm{in}_{#2,#3}^{#1}}
\def\KBC#1{\left(#1\right)'}
\def\Lam#1#2{\Lambda #1.\, {#2}}
\def\lam#1#2{\lambda #1.\, {#2}}
\def\l{{\pi_1}}

\def\N{\mathcal{N}}
\def\pairing#1#2{\langle #1,\, #2\rangle}
\def\plusetatop#1{{+.{\eta_{top#1}}}}
\def\prts{\KBC{\lambda^2\beta\eta\pi*}}
\def\r{{\pi_2}}
\def\R{\mathcal{R}}
\def\RED{\mathsf{RED}}
\def\reta{\overline{\eta}}
\def\rSP{\overline{SP}}
\def\rtc{\stackrel{*}{\to}}
\def\rts{\KBC{\lambda\beta\eta\pi*}}
\def\S{\mathcal{S}}
\def\SPtop#1{SP_{top#1}}
\def\sr{\left[\,\vec{x}:=\vec{u}\,\right]}
\def\SN{\mathcal{SN}}
\def\Sum{\lambda^{\top,\to,\times,+}}
\def\T{\mathcal{T}}
\def\tt#1{\mathop{T(#1)}}
\def\wfi#1{\mathrm{WFI}\left(#1\right)}
\begin{document}


\begin{abstract}For the lambda-calculus with surjective pairing and
  terminal type, 
  Curien and Di Cosmo were inspired by Knuth-Bendix completion, and introduced a confluent rewriting system that (1) extends the
  naive rewriting system, and (2) is stable under contexts. 
  The rewriting system has (i) a rewrite rule ``a term of a terminal type rewrites to a term constant $*$, unless the term is not $*$,'' 
  (ii) rewrite rules for the extensionality of  function types and product types, and rewrite rules mediating (i) and (ii). Curien and Di Cosmo supposed that
  because of (iii), any reducibility method cannot prove the strong normalization~(SN) of Curien-Di Cosmo's rewriting system, and 
  they left the  SN open.
  By relativizing Girard's reducibility method to the $\c$-free terms, we
  prove SN of their rewriting,  and SN of the extension by polymorphism.
 The relativization works because: for any SN term $t$, and  for any variable $z$ of terminal type  not occurring in $t$,
$t$ with all the occurrences of $\c$ of terminal type replaced by the variable $z$ is SN. 
\end{abstract}
\maketitle
\section{Introduction}\label{sec:introduction}

Equational theories for terminal types, unit types, singleton types are useful in mathematics and computer science:
 \begin{itemize}
 \item Coherence problem of cartesian closed category~\cite{Min77,Min79,Min92,MacLane82,LamSco}.
  
  \item An extension LF$^{\Sigma,1}$~\cite{Sarkar} of LF~\cite{HHP93} by dependent sum types and the type 1  for the empty context.
 \item Useless code elimination~\cite{UEC,KobaUEC99} 
 \item Proof irrelevant types~\cite{ACP}.
 \item Higher-order unification for a proof assistant system \texttt{Agda}~\cite{AP}.  
 \texttt{Agda}~\cite{Norell} supports $\Sigma$-types in form of records with associated $\eta$-equality in its general form.
 \end{itemize}
 We study the extensional $\lambda$-calculus $\lambda\beta\eta\pi*$ with surjective pairing and unit types.
It is an equational theory useful to solve the coherence problem of cartesian closed category. 
The equational theory $\lambda\beta\eta\pi*$ 
is decidable.  As we see below, typical proofs of the decidability employ, more or less, the following two methods:
\begin{itemize}
\item
Tait's reducibility methods to prove the strong normalization~(SN, for short) of rewriting relations;
Here, SN states that there is no infinite sequence of the rewriting relation.
Variants of Tait's reducibility method include reducibility candidate method~\cite{GTL} and computability closure~\cite{Blanqui}. 

\item  Logical relation methods. For the historical account, see \cite{HRR}. 
\end{itemize} 
In both methods, by induction on types, we define a family $\{P_{\varphi}\}_{\varphi}$ of sets  of terms,  indexed by all types $\varphi$s. Here 
\begin{align*}
(t_{1},\ldots,t_{n})\in P_{\varphi\to\psi}:\iff\forall (s_1,\ldots,s_{n}) \in P_{\varphi}.\, \left( (t_{1} s_{1},\ldots,t_{n} s_{n}) \in P_{\psi}\right).
\end{align*}
 Then we carry out an induction on terms to prove the target property. 
Logical relations more fit to semantical problems~\cite{SalvatiniWalukiewicz} of $\lambda$-calculi.

We list proofs of the decidability of the equational theory $\lambda\beta\eta\pi*$.
\begin{itemize}
\item Type-directed expansions. See \cite{Min77,Min79,hagiya,Cubric,A93,DK94,JG95,Lindley07}, to cite a few.
The SN proof of the type-directed expansion in \cite{JG95} is as follows: 
They first restricted the places of terms to be replaced,  
proved the SN of such restricted rewriting system
by a reducibility method, and then derived the SN of the type-directed expansion.

\item Sarkar's algorithm.
The extension LF$^{\Sigma,1}$ corresponds to $\lambda\beta\eta\pi*$. 
Sarkar~\cite{Sarkar}  studied LF$^{\Sigma,1}$ by the standard techniques of \cite{HarperPfenning}.
For LF$^{\Sigma,1}$, to give a type-checking algorithm,  Sarkar~\cite{Sarkar} 
provided an decision algorithm of the definitional equality. For the decision algorithm,
he proved the  completeness for equality by a Kripke logical relation, 
 the soundness of the algorithm and the existence of canonical forms in LF$^{\Sigma,1}$.

\item Normalization-by-evaluation~(\cite{CDS98,BDF04,ACD}, to cite a few). From a given term $t$, we obtain a normal form $v$ judgmentally equal to $t$, by evaluating $t$ and then by reification it. 
\cite{BDF04}~(\cite{ACD}, resp.) used Grothendieck logical relation~(Kripke logical relation, resp.) between well-typed terms $t$ and
semantic objects $d$, which for base types expresses that $d$ reifies to a normal
form $v$ judgmentally equal to $t$.
\item  A translation that incorporates type-directed expansions by type-indexed functions on
 terms. 
 The translation reduces the decidability of the equational theory $\lambda\beta\eta\pi*$ to that of the
 corresponding intensional equational theory~\cite{Gog05,Yokouchi}. It, however, turns out that this idea does
 not yield a decision procedure for the equational theory $\lambda^{2}\beta\eta\pi*$, which is the polymorphic extension of $\lambda\beta\eta\pi*$.
 \end{itemize}
 
In \cite{CD},  the decidability of the equational theory $\lambda\beta\eta\pi*$ are proved, much more based on rewriting technique~\cite{Baader}.
They first introduced a rewriting system that generates the equational theory $\lambda\beta\eta\pi*$, as follows: 
To the simply-typed $\beta\eta$-rewriting, we add the rewrite rule ``a term of type $\top$ rewrites to $*$ unless the term is not $*$,''  and then keep adding rewriting rules, like 
from a term rewriting system we obtain a confluent term rewriting system through Knuth-Bendix completion~\cite{Baader}.
For this extensional $\lambda$-calculus $\rts$ with surjective pairing and terminal type, 
 they proved the weak normalization of the rewriting system, and derived the confluence from it.
This rewriting system so directly depends on the rewriting technique.
The reducibility methods are not so flexible as rewriting rules.
Curien and Di Cosmo suggested no direct application of reducibility method proves SN of $\rts$.

We prove the SN of Curien-Di Cosmo's rewriting system $\rts$ by relativizing Girard's reducibility method to the $\c$-free terms. 
We introduce the reducibility predicates for $\rts$, apply them only for the set of $\c$-free terms,  derive the 
SN of all $\c$-free terms. 
To make our relativization argument handy, we introduce the non-Haussdorf Alexandrov topological space of terms for the rewriting, and interpret our argument.

The rest of paper is organized as follows: In the next section, we recall the definition of Curien-Di Cosmo's rewriting system~(Subsection~\ref{subsec:CD}), and explain how their rewriting system suggests relativization of reducibility method~(Subsection~\ref{subsec:CDC wanted}),
and uncover the essence of the relativization by using Alexandrov topological space~\cite{StoneSpaces}, in (Subsection~\ref{subsec:alex}).
In Section~\ref{sec:CD}, we prove SN of  $\rts$.  In Section~\ref{sec:polymorphism}, we prove  SN of $\prts$, the extension by the
polymorphism.
In Section~\ref{subsec:expansion}, we comment type-directed expansions as related work for Curien-Di Cosmo's rewriting. 

The preliminary version of this paper appeared as \cite{Akama17}.  
\section{Preliminary}
\subsection{Curien and Di Cosmo's rewriting system based on eta-reduction}\label{subsec:CD}

We recall the equational theory $\lambda\beta\eta\pi*$ from \cite{CD}.

Types are built up from the distinguished type constant $\top$, and type
variables, by means of the product type $\varphi\times \psi$ and the
function type $\varphi \to \psi$. Terms are built up from the
distinguished term constant $*^\top$ and term variables
$x^\varphi,y^\varphi,\ldots, x^\psi,y^\psi,\ldots$, by means of
$\lambda$-abstraction $(\lam{x^\varphi}{t^\psi})^{\varphi\to\psi}$, term
application $(u^{\varphi\to\psi} v^\varphi)^\psi$, pairing
$\pairing{u^\varphi}{v^\psi}^{\varphi\times\psi}$, left-projection $(\l
t^{\varphi\times\psi})^\varphi$, and right-projection $(\r
t^{\varphi\times\psi})^\psi$. The superscript represents the
type. The superscript is often omitted. The
set of free variables of a term $t$ is denoted by $\FV(t)$. The
equational theory $\lambda\beta\eta\pi*$ consists of the following
axioms:
\begin{align*}
 &(\beta)\quad &(\lam{x}{u})v   &=u[x:=v].\\
 &(\l)         &\l\pairing{u}{v}&=u.  &(\r)\quad         &\r\pairing{u}{v}=v.\\
 &(\eta)       &\lam{x}{tx}     &=t, \quad\mbox{($x\in\FV(t)$.)}\\
 &(SP)     &\pairing{\l u}{\r u}&=u.\\
 &(c)          &s^\top&=*^\top. 
\end{align*}
By the last equality,  the type $\top$ corresponds to the singleton. The
singleton does to the terminal object of a cartesian closed
category~(CCC for short). So $\top$
is called the \emph{terminal type}.

By orienting the equational axioms $(\beta),(\l),(\r),(\eta), (SP)$ left to right, we obtain rewrite rule schemata. Let $(T)$ be a rewrite rule schema
$s^\top\to \c^\top$ $(s^\top \not\equiv\c^\top$). Here for terms $t$ and $s$, we write $t\equiv s$, provided that by renaming bound variables, $t$ becomes identical to $s$.
Let $\to$ be the closure of these rewrite rule schemata by contexts.
By abuse of notation, we write $\lambda\beta\eta\pi*$ for a so-obtained rewriting system. The reverse of $\to$ is denoted by $\leftarrow$.  
$\rtc$ is the reflexive, transitive closure of $\to$. Let us abbreviate confluence by CR.

 The rewriting system
$\lambda\beta\eta\pi*$ is \emph{not} CR, as follows: In each line of the following, 
$x$ and $y$ are variables, and  it is not the
case that for the leftmost term $t_1$ and the rightmost $t_2$, there is a term $t_0$ such that $t_1\stackrel{*}{\to} t_0 \stackrel{*}{\rightarrow} t_2$\,:
\begin{align}
y^{\varphi\to \top}\leftarrow &\lam{x}{(y x)^\top}\to\lam{x}{*^\top}, \nonumber \\
x\leftarrow &\pairing{(\l x)^\top}{(\r x)^\top}\to\to
 \pairing{*^\top}{*^\top},  \label{SPpar}\\
\lam{x^\top}{y *}\leftarrow&\lam{x^\top}{y x^\top}\to y^{\top\to\varphi},\nonumber \\
\pairing{\l x}{*}\leftarrow&\pairing{(\l x)^\varphi}{(\r x)^\top}\to x^{\varphi\times\top}, &&\nonumber \\
\pairing{*}{\r x}\leftarrow&\pairing{(\l x)^\top}{(\r x)^\varphi}\to x^{\top\times\varphi}.&&\nonumber
\end{align}
The behavior of the rewrite rule schemata ($\gentop$) is not so simple as it looks like. The rewrite relation $\to_{\beta\eta\pi\l\r SP}$ is CR~\cite{Pottinger81}.
In the type-free setting, $\to_{\beta SP}$ is not CR~\cite{K80}.  In dependent type theories such as \texttt{Agda}, the unit type~(=terminal type) is important in relation to the record type, but
in the presence of the unit type, the type-checking is rather difficult; Not all subterms has a type label as our terms. So, we should  infer the type of the term before we apply the equational axiom $(c)$ to cope with a typing rule such as ``$M$ has a type $A$ whenever $M$ has a type $B$ such that $A$ is equal to $B$.''

For the equational theory $\lambda\beta\eta\pi*$, Curien and Di Cosmo, inspired by completion of term rewriting systems,
introduced a rewriting system $\rts$ in \cite{CD}. First they inductively defined the types \emph{``isomorphic to''} the terminal type $\top$ and
the \emph{canonical terms} of such types.
\begin{definition}[ $\rts$ ]  \label{def:isot}
\begin{itemize}
 \item $\top$ is ``isomorphic to'' $\top$ and the canonical term of $\top$
       is
 $*^\top$.
 \item Suppose $\varphi$ is a type and $\tau$ is a type ``isomorphic to''
       $\top$. Then the type $\varphi\to\tau$ is ``isomorphic to''
       $\top$ and  the canonical term
       $\c^{\varphi\to\tau}$ of $\varphi\to\tau$ is $\lam{x^\varphi}{\c^\tau}$.
       
 \item If each type $\tau_i$ is ``isomorphic to'' $\top$ ($i=1,2$), then
the type       $\tau_1\times\tau_2$ is ``isomorphic to''      $\top$  and
       the
       canonical term $\c^{\tau_1\times\tau_2}$ of $\tau_1\times\tau_2$  is
       $\pairing{\c^{\tau_1}}{\c^{\tau_2}}$.
\end{itemize}
 The set of types ``isomorphic to'' $\top$ is denoted by $\isomorphicT$.
    Whenever we write $\c^\varphi$, we tacitly assume $\varphi\isot$.
The canonical terms are not directly related to `the canonical forms of \cite[Sect.~8.1]{Sarkar}.

 The rewrite relation $\to$ of the rewriting system $\rts$ is defined by
the rewrite rule schemata obtained from the first five equational axioms
$(\beta),(\l), (\r),(\eta)$, and $(SP)$ of $\lambda\beta\eta\pi*$ by
orienting left to right, and the following four rewrite rule schemata:
 \begin{align*}
&(\gentop)         &u^\tau&\to \c^\tau, &\mbox{($u$ is not canonical.)}\\
&(\etatop)\quad & \lam{x^\tau}{t \c^\tau} &\to t, &
\mbox{($x\notin\FV(t)$.)}\\
&(\SPtop1)     &\pairing{\l u}{\c^\tau}&\to u,&\mbox{($u$ has type
  $\varphi\times \tau$.)}\\
  &(\SPtop2)     &\pairing{\c^\tau}{\r u}&\to u, &\mbox{($u$ has type
  $\tau\times \psi$.)}
 \end{align*}
The first rule~($\gentop$) schema that generates a canonical term $*^\tau$ is called ``gentop'' in \cite{CD}.
    \end{definition}

In \cite{CD}, the rewriting system $\rts$ is proved to be CR and weakly normalizing,
by using an ingenuous lemma for abstract reduction systems.  $\rts$ is non-left-linear and has a rewrite rule schema with side
conditions. We cannot apply criteria for CR of left-linear
(higher-order) term rewriting system based on closed condition of (parallel) critical pairs~(e.g., \cite{toyama,VO97}).
$\beta\eta\etatop\gentop$-reduction is the \emph{triangulation}~\cite{VOZ} of $\beta\eta\gentop$-reduction, and thus CR by \cite[Corollary~2.6]{VOZ}. However, $\rts$ is not a triangulation of the rewriting system $\lambda\beta\eta\pi*$;
As we see \eqref{SPpar}, $\gentop$-rule schema rewrites the one-step reduct $u^{\top\times\top}$ of $\pairing{\l u}{\r u}$  to the \emph{two-step} reduct of $\pairing{\l u}{\r u}$.
 This does not fit to the definition of the triangulation.

\subsection{Rewrite rule schema $(\etatop)$, and relativized reducibility method to the $\c$-terms}\label{subsec:CDC wanted}

All variations~(e.g., reducibility candidate method~\cite{GTL}, {computability closure}~\cite{Blanqui}) of
Tait's reducibility method uses reducibility predicates. The reducibility predicates for $\rts$ are as usual:

By an \emph{atomic type}, we mean the distinguished type constant $\top$ or a type variable.

\begin{definition}\label{cd:reducibility}
\begin{description}
\item[$(a)$]
  A term of an atomic type is \emph{reducible}, if the term is SN.  
  \item[$(\times)$]
  A term $t^{\varphi\times\psi}$ is \emph{reducible}, if so are
  $(\l t)^\varphi$ and $(\r t)^\psi$.  
  
  \item[$(\to)$]
  A term $t^{\varphi\to\psi}$ is
  \emph{reducible}, if for any reducible term $u^\varphi$, $(t u)^\psi$ is
  reducible.
  \end{description}
\end{definition}
Let $RED_{\varphi}:=\{t^{\varphi} \mid t^{\varphi} \mbox{is reducible}\}$.
All variations of reducibility method require  to show a \emph{key statement} 
\begin{align}\forall u^\varphi \in RED_{\varphi}\ \left( v^\psi [x^\varphi :=u^\varphi ]\in RED_{\psi}\right)
\implies\lam{x}{v}\in RED_{\varphi\to\psi}.\label{key statement for abstraction}
\end{align}

The rewrite rule schema $(\etatop)$, however, causes the difficulty to prove the key statement~\eqref{key statement for abstraction}, as follows~\cite{CD}:
In the reducibility candidate method \cite{GTL},  
an available auxiliary property is that, a term $t u$ is reducible, as soon as $s$ is reducible for all reducts
$s$ of $t u$. So the proof of the key
statement amounts to the proof that  all reducts of
$(\lam{x}{v}) u$ are reducible. The rewrite rule schema
($\etatop$) can rewrite $(\lam{x}{v})u$ to $(v' u)$ which is not
$v[x:=u]\equiv v$. 
The standard argument indeed proves
the following statement~(Lemma~\ref{lem:weakened sufficient condition T}~\eqref{lambda}):
\begin{align} 
\left\{\begin{array}{l}\forall u\in RED_{\varphi}\ \left( v[x:=u]\in RED_{\psi}\right)\mbox{ and}\\
 \left(v \equiv (v' *^{\top}),\ x\notin\FV(v')\ \implies\ v' \in RED_{\top\to\psi}\right)\end{array}\right\}
\implies\lam{x}{v}\in RED_{\varphi\to\psi}, \label{naive}
\end{align}
This immediately implies
\begin{align} 
v^\psi \in F\; \&\; \forall u\in RED_{\varphi}\ \left( v[x^\varphi :=u^\varphi ]\in RED_{\psi}\right)\implies\lam{x}{v}\in RED_{\varphi\to\psi}\cap F,\label{naive1}
\end{align}
where
\begin{definition}\label{star-free}Let $t$ be a term  of $\rts$.
 $t$ is called \emph{$\c$-free}, if the term constant $*^\top$ does not
 occur in $t$.  Let $F$ be the set of $\c$-free terms.
 Let $T$ be the set of terms and $\SN$ be the set of SN terms.  
\end{definition}
The lemma~\eqref{naive1} suggests to  split  $T\subseteq\SN$ into two statements
\begin{align}
&  F\subseteq\SN\implies T\subseteq\SN, \\ 
&  F\subseteq\SN, \label{dag}
\end{align}
and to prove $F\subseteq\SN$ by employing $\{\RED_{\varphi}\cap F\mid\varphi\ \mbox{is a type}\}$.

\subsection{Essence of reducibility predicate relative to $*$-free terms}\label{subsec:alex}

To prove $F\subseteq\SN$, we prove the following relativization of the key statement~\eqref{key statement for abstraction} to $F$.
\begin{align} 
v^\psi \in F\; \&\; \forall u\in (RED_{\varphi}\cap F)\,\left( v[x:=u]\in RED_{\psi}\cap F \right)\implies\lam{x}{v}\in RED_{\varphi\to\psi}\cap F.\label{naive2}
\end{align}
This follows from lemma~\eqref{naive1}, if for a function $f_{v^\psi,x^\varphi}(u^\varphi):=v^\psi[x^\varphi:=u^\varphi]$,
\begin{align}
       f_{v^\psi,x^\varphi} \left(RED_\varphi\cap F\right) \subseteq  RED_\psi \cap F\  \implies\  f_{v^\psi,x^\varphi} \left(RED_\varphi\right) \subseteq  RED_\psi.\label{important lemma}
\end{align}
So, relativizing reducibility method to $F$ is introducing a topology to $T$ such that
\begin{enumerate}
\item $f_{v^\psi,x^\varphi}:T\to T$ is continuous,
\item $\overline{\RED_{\varphi}\cap F}=\RED_{\varphi}$ $(\forall\varphi)$, where $\overline{(\ )}$ is the closure operation, and 
\item $F\subseteq\SN\implies T\subseteq\SN$.
\end{enumerate}
The pair of the first two implies \eqref{important lemma}, because
\begin{align*}
\RED_\psi=\overline{\RED_\psi\cap F}\ \supseteq \ \overline{f_{v,x}(\RED_\varphi\cap F)}\ \supseteq \  f_{v,x}( \overline{\RED_\varphi\cap F})=f_{v,x}( \RED_\varphi).
\end{align*}

 A topological space $X$ is called  \emph{Alexandrov}, if  there is a preorder $\le$ such that the closed sets are exactly the  upwardly closed sets.
The Alexandrov topological space induced by a preorder $\le$ is denoted by $\tt{\le}$.
Let $\le$ and $\sqsubseteq$ be preorders.
A function $f : \tt{\le} \to \tt{\sqsubseteq}$ is continuous, if and only if $f$ preserves the preorders.

For  Curien-Di Cosmo's rewriting $\to$, we consider the Alexandrov topological space $T=\tt{\rtc}$ of terms.
For every  $A\subseteq\tt{\rtc}$, the closure $\overline{A}$ is $\{t'\mid \exists t\in A.\,(t\rtc t')\}$.
This topology satisfies the above-mentioned three conditions:
\begin{enumerate}
\item
 $f_{v,x}$ is continuous, because
$u\to u'\implies v[x:=u]\rtc v[x:=u']$.

\item $\overline{RED_{\varphi}\cap F}=RED_{\varphi}$ $(\forall \varphi)$, because
\begin{align}
\mbox{A variable $z^\top$ does not occur in $t\in \SN$} \implies t[\c^\top := z^\top ] \in \SN, \label{essence}
\end{align}
as we will see in the proof of Lemma~\ref{lem:rel}~\eqref{red rel}.

\item $F\subseteq\SN\implies \tt{\rtc}\subseteq\SN$, by
$\overline{F}=\tt{\rtc}$ and the closedness of $\SN$.
\end{enumerate}
The Alexandrov topology uncover the essence of the reducibility method relativized to the $\c$-terms, that is, the property~\eqref{essence}.

The property~\eqref{essence} is also the essence of SN proof of  the polymorphic extension $\prts$ of $(\lambda\beta\eta\pi*)'$.
Girard proved the polymorphic $\lambda$-calculus $\lambda^{2}$ by employing the candidates of reducibility for all types~\cite{GTL}.
We will prove the SN of the Curien-Di Cosmo-style polymorphic $\lambda$-calculus $\prts$ by relativizing the reducibility candidate method to the $\c$-terms, as we did SN of 
$(\lambda\beta\eta\pi*)'$ by relativizing the reducibility method to the $\c$-terms.
For the SN proof of $\prts$, in Definition~\ref{def:RC}, we additionally require the following property to each reducibility candidate $\R$ of each type $\varphi$:
\begin{property}\label{property}
\begin{enumerate}
\item
 If $\c^\varphi$ is defined, then $\c^\varphi\in\R$; and
 
 \item 
For any $t\in \R$, and  for any variable $z^\top$  not occurring in $t$,
$t[\c^\top:=z^\top]\in\R$. 
\end{enumerate}
\end{property}
Then $\overline{\R\cap F}=\R$.

\section{SN proof  by relativized reducibility method} \label{sec:CD}

In our SN proofs, we will use a \emph{well-founded induction on a well-founded relation}.
A \emph{well-founded} relation is, by definition, $\A=(A,\succ)$ such that $\emptyset\ne\ \succ\ \subseteq A\times A$ and there is no infinite chain 
$a\succ a' \succ a'' \succ\cdots$. The \emph{well-founded induction} on a well-founded relation $\A=(A,\succ)$ is, by definition,
\begin{align*}
\wfi{\A}:\ \ \forall P\subseteq A\left[ \begin{array}{l}\forall x\in A\bigl(\forall x'\left(x\succ x'\Rightarrow x'\in P\right)
\implies x\in P\bigr)\\
\implies \forall x \in A(x\in P)\end{array}\right].
\end{align*}
We call the subformula $\forall x'\left(x\succ x'\Rightarrow x'\in P\right)$ the \emph{WF induction hypothesis}. 
For $n\ge1$ well-founded relations $\A_{i}=(A_{i},\, \succ_{i})$ $(i=1,\ldots,n)$, we define a binary relation 
$$\A_{1}\#\cdots\# \A_n=\left(A_{1}\times\cdots\times A_{n},\ \succ_{1}\#\cdots\#\succ_{n}\right)$$ by:
$(x_{1},\ldots,x_{n})\ \left(\succ_{1}\#\cdots\#\succ_{n}\right)\ (y_{1},\ldots,y_{n})$,
if there exists $i$ such that $x_{i}\succ_{i} y_{i}$ but $x_{j}=y_{j}$ $(j\ne i)$.
Then $\A_{1}\#\cdots\# \A_{n}$ is  a well-founded relation.

If the redex of $t\to t'$ is $\Delta$, we write
$t\stackrel{\Delta}{\to} t'$. Below, ``$\subseteq$''  reads ``is a subterm
occurrence of .''

Following \protect{\cite{GTL}}, we consider:

\begin{definition}[Neutral]\label{def:neutral}
 A term is called \emph{neutral} if it is not of the form
$\pairing{u}{v}$ or $\lam{x}{v}$. \end{definition}

  We state and prove four properties (CR0), (CR1), (CR2) and (CR3) of the
reducibility~(Definition~\ref{cd:reducibility}). Girard verified the last three to prove the SN of $\beta\l\r$-reduction in \cite{GTL}.
  
 \begin{lemma}\label{6.2.3} 
  \begin{description}
   \item[(CR0)]\ If $\c^\varphi$ is defined, then $\c^\varphi$ is reducible.
   \item[(CR1)]\ If $t^\varphi$ is reducible, then $t^\varphi$ is SN.
   \item[(CR2)]\ if $t^\varphi$ is reducible and $t^\varphi\to t'^\varphi$, then $t'^\varphi$ is reducible.
   \item[(CR3)]\ if $t^\varphi$ is neutral, and $t'$ is reducible whenever $t^\varphi\to t'^\varphi$,
	      then
	      $t^\varphi$ is reducible.
\end{description}
 \end{lemma}
 To prove Lemma~\ref{6.2.3}, we first note the following:
 \begin{lemma} By (CR0) and
  (CR3), we have  \begin{description}
   \label{varred}
      \item[(CR4)]\ If  $t^\varphi$ is  a variable, then $t^\varphi$ is reducible.
\end{description}
 \end{lemma}
  \begin{proof}Let $t\to t'$. Then $t'$ is canonical, since $t$ is variable. By (CR0), $t'$ is reducible. 
By (CR3), $t$ is too.\end{proof}

   \medskip
   \noindent
   \emph{Proof of Lemma~\ref{6.2.3}}. By induction on $\varphi$.
  \begin{description}
   \item[$\varphi$ is atomic.]
   
   (CR0) $\c^\varphi$ is $\c^{\top}$, and SN. So $\c^\varphi$ is reducible.

	 (CR1) is clear. (CR2) As $t$ is SN,  so is
	 every reduct $t'$ of $t$.
	      (CR3)
If all reducts of $t$ are SN, then $t$ is SN.

   \item[$\varphi=\varphi_{1}\times\varphi_2$.]
   
    (CR0) As $\c^{\varphi_1\times\varphi_2}$ is a normal form $\pairing{\c^{\varphi_1}}{\c^{\varphi_2}}$, the reduct of $\pi_{i}\c^{\varphi_1\times\varphi_2}$ is $\c^{\varphi_{i}}$, which is reducible by \ih0. By \ih3 for $\varphi_{i}$, $\pi_{i} \c^{\varphi_1\times\varphi_2}$ is reducible. Hence $\c^{\varphi_1\times\varphi_2}$ is reducible.

\medskip
	 (CR1) Suppose that $t$ is reducible. Then $\pi_{i} t$ is reducible. By \ih1 for $\varphi_{i}$, $\pi_{i} t$ is SN.  So $t$ is SN.

\medskip
	 (CR2) If $t\to t'$, then $\pi_{i} t\to \pi_{i} t'$. As $t$ is reducible by hypothesis, so are $\pi_{i} t$. By  \ih2 for $\varphi_{i}$,  $\pi_{i} t'$ is reducible, and so $t'$ is reducible.

\medskip
	 (CR3)  Let $\pi_i t \stackrel{\Delta}{\to} s$. We have two cases.
	 \begin{enumerate}
	 \item $\Delta\equiv \pi_i t$ and
	 $s\equiv\c^{\varphi_{i}}$: By \ih0 for
	 $\varphi_{i}$, $s$ is reducible. 
	 \item Otherwise,
	 $s\equiv \pi_i t'$ for some $t'$ such that $t\to t'$.   $s$ is reducible, because $t'$ is reducible by the
	 hypothesis.
	   $\pi_i t$ is neutral, and all the terms $s$ with
	 $\pi_{i} t\to s$ are reducible. By \ih3 for $\varphi_{i}$,  $\pi_i
	 t$ is reducible. Hence $t$ is reducible.
	\end{enumerate}

   \item[$\varphi=\varphi_{1}\to\varphi_2$.]
   
   (CR0) Let $u$ be a reducible term of type $\varphi_{1}$. By \ih1 for $\varphi_{1}$, $u$ is SN.
So $\wfi{\left(\left\{u^{\varphi_1}\mid u^{\varphi_1}\ \mbox{is reducible}\right\},\ \to\right)}$ is available, where $\to$ is the rewrite relation.
    We will verify that $\c^{\varphi_{1}\to\varphi_{2}} u$ is reducible. Suppose $\c^{\varphi_{1}\to\varphi_{2}} u\stackrel{\Delta}{\to} s$. 
   As  $\c^{\varphi_{1}\to\varphi_{2}}$ is in normal form, we have two cases.
   \begin{enumerate}
   \item $\Delta\equiv \c^{\varphi_1\to\varphi_2}u$: Then $s\equiv \c^{\varphi_{2}}$ is reducible by \ih0 for $\varphi_{2}$.
    \item Otherwise, $s\equiv \c^{\varphi_{1}\to\varphi_{2}}u'$ with $u\to u'$. Then $u'$ is reducible
		by \ih2 for $\varphi_{1}$. So, by the WF
		induction hypothesis, $s\equiv \c^{\varphi_{1}\to\varphi_{2}} u'$ is
		reducible.
\end{enumerate}
In any case,  the neutral term $\c^{\varphi_{1}\to\varphi_{2}} u$
	 rewrites to reducible terms only. By \ih3 for $\varphi_2$,
	 $\c^{\varphi_{1}\to\varphi_{2}} u$ is reducible. So $\c^{\varphi_{1}\to\varphi_{2}}$ is reducible.
	 
\medskip
	 (CR1) By \ih4, a variable $x^{\varphi_{1}}$ is reducible. So $t x$ is
	 reducible.
 Hence $t$ is SN.

\medskip
	 (CR2) Let $u$ be a reducible term of
	 type $\varphi_{1}$. Then $t u$ is reducible and $t u \to t' u$. By the
	 \ih2 for $\varphi_2$, $t' u$ is reducible. So $t'$ is
	 reducible.

\medskip
	 (CR3) Assume that $t$ is neutral and that all the reducts $t'$ of $t$ are
	 reducible. Let $u$ be a reducible term of type $\varphi_{1}$. By
	 \ih1 for $\varphi_{1}$, $u$ is SN. 
	 So by
	 $\wfi{\left(\left\{u^{\varphi_1}\mid u^{\varphi_1}\ \mbox{is reducible}\right\},\ \to\right)}$, we will verify that $t u$ is reducible.

	 Suppose $t u\stackrel{\Delta}{\to} s$. We will show that $s$ is reducible. As $t$ is neutral, we have three cases.
	 \begin{enumerate}
	  
	  \item $\Delta\equiv t u$:  Then,  $s\equiv \c^{\varphi_2}$  is reducible, by \ih0 for $\varphi_2$.   
		
	  \item $\Delta\subseteq t$: Then, $s\equiv t' u$ with $t\to t'$. $t' u$ is reducible,
		because $t'$ is by the assumption, 

	  \item Otherwise, $s\equiv t u'$ with $u\to u'$. Then, $u'$ is reducible
		by \ih2 for $\varphi_{1}$.
		So, by the WF
		induction hypothesis, $t u'$ is
		reducible.
	 \end{enumerate}

	 In any case,  the neutral term $t u$
	 rewrites to reducible terms only. By \ih3 for $\varphi_2$,
	 $t u$ is reducible. So $t$ is reducible.  This
completes the proof of Lemma~\ref{6.2.3}.\qed
\end{description}

The following lemma suggests that the rewrite rule schemata $(\etatop), (\SPtop{})$ require the relativization of 
the key statement for pairing and that of $\lambda$-abstraction in the reducibility method.

\begin{lemma}\label{lem:weakened sufficient condition T}
   \begin{enumerate}
   \item \label{pairing}  Let $u^\varphi,v^\psi$ be any terms.
   $\pairing{u^\varphi}{v^\psi}$ is reducible, provided that
	 \begin{enumerate}
	  \item \label{prod1}
		$u$ and $v$ are both reducible;

	  \item \label{prod2}if
		$u\equiv \l w$ and $v \equiv \c^\psi$,
		then $w$ is reducible; and 
	  \item \label{prod3} if
		$v\equiv \r w$ and $u \equiv \c^\varphi$,  then $w$ is reducible.
\end{enumerate}

   \item \label{lambda} Let $v^\psi$ be any term.
   $\lam{x^\varphi}{v^\psi}$ is reducible, provided that
	 \begin{enumerate}
	  \item \label{arrow1}
 $v^\psi[x^\varphi:=u^\varphi]$ is reducible for every reducible, \emph{possibly non-$\c$-free} term $u^\varphi$; and

	  \item \label{arrow2} if $v\equiv w^{\varphi\to\psi} \c^\varphi$
		and $x^\varphi\notin \FV(w^{\varphi\to\psi})$,
		then $w^{\varphi\to\psi}$ is reducible.
		      \end{enumerate}
  \end{enumerate} \end{lemma}

  \begin{proof}
  (1) By the premise  and
  (CR1), $u$ and $v$ are both SN. We can use
  \begin{align}\wfi{\left(\left\{u^{\varphi}\mid u^{\varphi}\ \mbox{is reducible}\right\},\ \to\right)\#\left(\left\{v^{\psi}\mid v^{\psi}\ \mbox{is reducible}\right\},\ \to\right)} \label{wf:prod}
  \end{align}
  where $\to$ is the rewrite relation.
We will verify
  that $\l \pairing{u}{v}$ is reducible. Let $\l \pairing{u}{v}\stackrel{\Delta}{\to}  s$. We will prove that $s$ is reducible, by case analysis.  We will exhaust the positions of the redexes $\Delta$ in $\pi_1\pairing{u}{v}$ from left to right, and the rewrite rule schemata of $\stackrel{\Delta}{\to}$.
 We have eight cases.
  \begin{enumerate}
   \item  $\Delta\equiv \l\pairing{u}{v}$ is a redex and $s\equiv\c^{\varphi}$: Then
	  $s$ is reducible by (CR0).

   \item $\Delta\equiv \l\pairing{u}{v}$ is a redex of the rewrite rule
	 ($\l$) and $s\equiv u$: Then $s$ is reducible by the hypothesis~\eqref{prod1}.

   \item $\Delta\equiv \pairing{u}{v}$ is a redex and
	$s\equiv \l (\c^{\varphi\times\psi})$: Then
	  $\c^{\varphi\times\psi}$ is reducible by (CR0).
By the definition of the reducibility for the product type, 
	 $s\equiv\l (\c^{\varphi\times\psi})$   is reducible.

   \item $\Delta\equiv \pairing{u}{v}$ is a redex of ($SP$) and
	 $s\equiv \l w$: Then
	  $u\equiv\l w$ and $v\equiv \r w$. $s\equiv\l w$ is
	 reducible by the hypothesis~\eqref{prod1}
	 
  \item $\Delta\equiv \pairing{u}{v}$ is a redex of ($\SPtop1$) and
	 $s\equiv \l w$: Then
	  $u\equiv\l w$ and $v\equiv \c^\psi$. $s\equiv\l w$ is
	 reducible by the hypothesis~\eqref{prod2}

   \item $\Delta\equiv \pairing{u}{v}$ is a redex of ($\SPtop2$) and
	 $s\equiv \l w$: Then
	  $v\equiv\r w$ and $u\equiv \c^\varphi$. $s\equiv\l w$ is
	 reducible by the hypothesis~\eqref{prod3}

   \item $\Delta\subseteq u$: Then $s\equiv \l \pairing{u'}{v}$ with
	 $u\to u'$.  $u'$ is
	 reducible by \eqref{prod1} and (CR2). By the WF
	 induction hypothesis,  $s\equiv \l \pairing{u'}{v}$ is reducible.

   \item $\Delta\subseteq v$: Then  $s\equiv \l\pairing{u}{v'}$ with
	 $v\to v'$. $v'$ is reducible by \eqref{prod1} and (CR2). By the WF
	 induction hypothesis,  $s\equiv \l \pairing{u}{v'}$ is reducible.
  \end{enumerate}
  In every case, the neutral term $\l\pairing{u}{v}$ rewrites to
  reducible terms only, and by (CR3), $\pi_1\pairing{u}{v}$ is reducible. We can similarly
  prove that $\r
  \pairing{u}{v}$ is reducible.  So $\pairing{u}{v}$ is reducible.

  \medskip
  (2) By (CR4), $x^{\varphi}$ is reducible. So  $v^{\psi}$ is, by the premise~\eqref{arrow1}.
  Let $u^{\varphi}$ be a reducible, possibly non-$\c$-free term.
    By
  (CR1),  both of $u,v$ are SN.  By the well-founded induction~\eqref{wf:prod}, 
    we will verify that $(\lam{x}{v})u$ is reducible.
     Assume $(\lam{x}{v})u\stackrel{\Delta}{\to} s$.  
     We will exhaust the positions of the redex $\Delta$ in $(\lam{x}{v})u$ from left to right, and the rewrite rule schemata of $\stackrel{\Delta}{\to}$.
         Then we have seven cases:
  \begin{enumerate}
   \item   $\Delta\equiv (\lam{x}{v}) u$ is a redex and $s\equiv\c^{\psi}$: Then
	  $s$ is reducible by (CR0).

 \item $\Delta\equiv (\lam{x}{v}) u$ is a redex of ($\beta$) and $s\equiv v[x:=u]$: Then $s$ is reducible by hypothesis~\eqref{arrow1}.

   \item $\Delta\equiv \lam{x}{v}$ is a redex and $s\equiv \c^{\varphi\to\psi} u$:  As $\c^{\varphi\to\psi}$ is reducible
	 by (CR0), so is $s$.
	 
 \item $\Delta\equiv \lam{x}{v}$ is a redex of ($\eta$) and $s\equiv v[x:=u]$:  Then,  this case is case~2.

   \item $\Delta\equiv \lam{x}{v}$ is a redex of ($\etatop$) and $s\equiv w u$ with $v \equiv w \c^\varphi$ and $x\notin\FV(w)$:  Then, since $w$ is reducible by hypothesis~\eqref{arrow2}, $s\equiv w u$ is reducible.

     \item $\Delta\subseteq v$ and $s\equiv (\lam{x}{v'}) u$ with $v\to v'$: Then, by (CR2), $v'$ is
	 reducible.  By the WF induction hypothesis,
$s\equiv(\lam{x}{v'}) u$ is reducible.
	 
\item $\Delta\subseteq u$ and $s\equiv (\lam{x}{v}) u'$ with $u\to u'$: Then, by (CR2), $u'$ is reducible.
   By the WF induction hypothesis,
$s\equiv (\lam{x}{v}) u'$ is reducible.
\end{enumerate}
  
In every case, the neutral term $(\lam{x}{v}) u$ reduces to reducible
terms only. So, by (CR3),  $(\lam{x}{v}) u$ is reducible. Hence $\lam{x}{v}$ is reducible.
 \end{proof}

 \begin{corollary}\label{cor:product} If $u^\varphi,v^\psi$ are reducible and $\c$-free, then so is $\pairing{u}{v}$.\end{corollary}
  
\begin{lemma}\label{lem:rel} 
 \begin{enumerate}
 \item \label{assert:term subst T} A $\c$-free term with the variables substituted by $\c$-free terms is $\c$-free.

 \item 
 \label{red rel} Suppose that $t$ is a term and $z^\top$ is a variable not occurring in $t$.
 Then
 \begin{enumerate}
 \item \label{assert red rel}
if $t$ is reducible, so is $  t[\c^\top:=z^\top]$. 
 
 \item \label{rel} $t [ \c^\top:=z^\top ]\rtc t$.
 \end{enumerate}
  \end{enumerate}
\end{lemma}
 \begin{proof}
 \eqref{assert:term subst T}  Trivial.
 
\eqref{red rel} 
\eqref{rel} is trivial.

\eqref{assert red rel} By induction on the type $\varphi$ of $t$. 
 \begin{itemize}
     \item $\varphi$ is atomic:

Assume that 
                       $t [ \c^\top:=z^\top ]$ is not reducible. By the definition, $t$ is SN but 
                       there are $  t [ \c^\top:=z^\top ] \equiv s_{0}, s_{1},s_{2},\ldots$ such that
                     $s_{i}\to s_{i+1}$. Then $s_{i}[ z^\top:=*^\top]\equiv s_{i+1}[ z^\top:=*^\top]$ or 
                     $s_{i}[ z^\top:=*^\top]\to s_{i+1}[ z^\top:=*^\top]$. The former happens if $s_i\to_\gentop s_{i+1}$ with the redex being $z^\top$. If $\{s_{i}[ z^\top:=*^\top]\}_{i}$ is finite, then
                     for any $i$ but finitely many, $s_{i}\to_\gentop s_{i+1}$. However,  $\to_\gentop $ is SN, because $\to_\gentop$ reduces the length of terms or the number of non-$*$ variables. 
  Since $z$ is a fresh variable,  $t\equiv s_{0}[z:=*]$. Hence, $t$ is not SN. This contradicts the reducibility of $t$.

    \item $\varphi=\varphi_1\to\varphi_2$:
  
  As $t^{\varphi_1\to\varphi_2}$ is reducible, $(t u)^{\varphi_2}$ is so for every reducible $u^{\varphi_1}$. 
  $z^\top$ does not occur in $t u$. So, by induction hypothesis on $\varphi_2$, $( t u ) [ \c^\top:=z^\top ] \equiv t [ \c^\top:=z^\top ]  u [ \c^\top:=z^\top ]$ is reducible.
  By \eqref{rel},  $u [ \c^\top:=z^\top ] \rtc u$. 
  So, $( t u ) [ \c^\top:=z^\top ] \rtc t  [ \c^\top:=z^\top ] u$. By (CR2), $t  [ \c^\top:=z^\top ] u$ is reducible.
 Hence $t [ \c^\top:=z^\top ]$ is reducible.
  
      \item $\varphi=\varphi_1\times \varphi_2$:
  
  As $t^{\varphi_1\to\varphi_2}$ is reducible, $\pi_i t^{\varphi_i}$ is so for each $i=1,2$.
  $z$ does not occur in any of $\pi_i t$. 
  So,  by induction hypothesis on $\varphi_i$, $(\pi_i t) [ \c^\top:=z^\top ] \equiv \pi_i (t [ \c^\top:=z^\top ] )$ is reducible. 
Hence $t[*^\top:=z^\top]$ is reducible.
\end{itemize}

\eqref{rel} Just contract each occurrence of $z^\top$ to $*^\top$.
\end{proof}

\begin{lemma}\label{key statement for arrow}Given a $\c$-free $v^\psi$. If $v[x^\varphi:=u^\varphi]$ is reducible for every reducible $\c$-free $u^\varphi$, then $\lam{x^\varphi}{v^\psi}$ is reducible and $\c$-free.\end{lemma}
\begin{proof}
 Let $w^\varphi$ be a reducible term. By Lemma~\ref{lem:rel}~\eqref{red rel}, there is a $\c$-free reducible term $u$ such that $u\stackrel{*}{\to} w$. By the premise, $v[x:=u]$ is reducible.  
 Because of $v[x:=u]\stackrel{*}{\to} v[x:=w]$, $(CR2)$ implies that $v[x:=w]$ is reducible. By Lemma~\ref{lem:weakened sufficient condition T}~\eqref{lambda}, $\lam{x^\varphi}{v^\psi}$ is reducible.
\end{proof}
 
 In the following two theorems, we use Lemma~\ref{lem:rel}.
  
 \begin{theorem}[Relativized Reducibility]\label{prop}
Assume that 
\begin{enumerate}
 \item \label{star free}
$t$ is a \emph{$\c$-free} term; 
\item a sequence of distinct variables  $x_1^{\varphi_1},\ldots, x_n^{\varphi_n}$
 contains all free variables of $t$; and 
\item \label{ured} $u_i^{\varphi_i}$ is reducible
 and $*$-\emph{free} $(i=1,\ldots,n)$.
\end{enumerate}
Then  $t[x_1^{\varphi_1},\ldots, x_n^{\varphi_n}:=u_1^{\varphi_1},\ldots, u_n^{\varphi_n}]$ is reducible.
\end{theorem}

\begin{proof} 
 By induction on $t$. By the premise~\eqref{star free}, $t$ is not the constant $\c^\top$. So, we have five cases.
 \begin{enumerate}
  \item $t\equiv x_i$: Then $t\sr\equiv u_i$ is reducible by the premise~\eqref{ured}.
	
  \item $t\equiv \pi_i w$ $(i=1,2)$:
	Then by induction hypothesis, 
$w\sr$ is reducible. So is each $\pi_i (w\sr)$. This term
	is  $(\pi_i w)\sr  \equiv t\sr$.

  \item $t\equiv \pairing{u}{v}$:  
	Then $t\sr\equiv \pairing{u\sr}{v\sr}$.  By  induction hypotheses,
	both $u\sr$ and $v\sr$ are reducible.
	By Lemma~\ref{lem:rel}~\eqref{assert:term subst T}, $u\sr$ and $v\sr$ are $\c$-free.
	By Lemma~\ref{lem:weakened sufficient condition T}~\eqref{pairing}, $t\sr$, that is, $\pairing{u\sr}{v\sr}$, is
reducible.

  \item $t\equiv w v$:
	Then by induction hypotheses $w\sr$ and $v\sr$ are
	reducible, and so (by definition) is $w\sr (v\sr )$; but this
	term is $t\sr$.

  \item $t\equiv \lam{y^\varphi}{w^\psi}$ with $y$  not  free in any $\vec x, \vec u$:
By induction hypothesis, for every reducible $\c$-free $u^\varphi$, we have a reducible term $w[\vec{x},y^\varphi:=\vec{u},u^\varphi]\equiv w[\vec{x}:=\vec{u}][y^\varphi:=u^\varphi]$.
	By Lemma~\ref{key statement for arrow},
	$\lam{y^\varphi}{w \sr}\equiv  t\sr$
	is reducible. 
 \end{enumerate}
 Hence we have established the relativized reducibility theorem.
\end{proof}

\begin{theorem}\label{cd:reducibility theorem}
  All terms of $\rts$ are reducible.
\end{theorem}

\begin{proof}
 Let $t$ be a term. Som variable $z^\top$ does not occurring in $t$. 
 By Lemma~\ref{lem:rel}~\eqref{rel}, there is a $\c$-free term $\tilde t$ such that $\tilde t\rtc t$.
 $\tilde t$ is reducible
  by
(CR4) and by Theorem~\ref{prop} with $u_i:=x_i$, the identity
substitution. As $\tilde t\rtc t$, (CR2) implies the reducibility of $t$.\end{proof}

 \begin{corollary}\label{thm:sncr}\rm
 $\rts$ satisfies SN.
\end{corollary}
\begin{proof}
By (CR1) and Theorem~\ref{cd:reducibility theorem}, every term of $\rts$ is SN. \end{proof}

We can define the extension of the equational theory $\rts$ by weakly extensional sum types,
and the extension of the Curien-Di Cosmo style rewriting system, and prove the SN by a relativized reducibility method~\cite[Appendix]{Akama17}.

\begin{remark}\label{sec:related work}\rm
In \cite{H70}~(\cite{S77}, resp.),  ordinal numbers are assigned to
typed $\lambda$-terms~(typed combinators, resp.) in order to prove SN
 of typed $\beta$-reduction~(typed combinatory
reduction, resp.). In \cite{B01}, cut-elimination procedure of a
deduction system is used to give an optimal upper bound of typed
$\beta\eta$-reduction.
But these two proofs seem not to generalize for SN of the rewriting system
$\rts$. In these two proofs, it is not the case that
(1) the ordinal number of $r\c^\tau$  is greater than that of $r$ and (2) the ordinal number of  the left-hand side $\lam{x^\tau}{t*^\tau}$ ($x\notin \FV(t)$) of the
rewrite rule schema $(\etatop)$ is greater than the ordinal number of the right-hand side $t$.

One may be curious about whether the higher-order recursive path ordering (HORPO for short)~\cite{jr} or the
General Schema~\cite{bjo}, could be
extended with surjective pairing and hence be used for proving SN of $\rts$. 
If there is a convenient translation of the rewrite rule schemata $(\gentop), (\eta_{top})$, and $(SP_{top})$ with
type-abstraction to an infinite simply-typed system, such that the
translation can also put all the rules of $\rts$ in the right kind of
format, it is possible that a HORPO-variant (with minimal symbol *) may
handle $\rts$. However, we need a new HORPO variant,
since the conventional ones are troubled with the  non-left-linear $(SP)$-rule \texttt{pair(p1(X), p2(X))->X}. There is no
type ordering that allows for the extraction of $\mathtt X$ from terms of smaller
type in general. The top rule ($\gentop$): $u^{\tau}\to \c^\tau$
$(\tau \isot, \ u\not\equiv \c^\tau)$ is also problematic for most
HORPO-variants. It could be handled
by using a variation of HORPO with minimal symbols, such as the one used in WANDA~\cite{kop}.
Here, WANDA is one of the most powerful automatic termination provers for higher-order rewriting.
\end{remark}
\section{SN proof  by relativized reducibility candidate method}\label{sec:polymorphism}

In \cite{CD}, an extension $\prts$ of $\rts$ by polymorphism is introduced.

Types are generated from type variables $X,Y,\ldots$
and the distinguished type constant $\top$ by means of the product type
$\varphi\times\psi$, the function type $\varphi\to\psi$, and $\all{X}\varphi$. 
  
Terms are built up similarly as the terms of $\lambda\beta\eta\pi*$, but we also consider the following two clauses:
 \begin{itemize}
  \item \emph{universal abstraction}: if $v^\varphi$ is a term, then so is $\left(\Lam{X}{v^\varphi}\right)^{\all{X}{\varphi}}$, whenever the type variable $X$ is not free in the
	type of a free variable of $v^\varphi$; and
  \item \emph{universal application}: if $t^{\all{X}{\varphi}}$ and
	$\psi$ is a type, then so is $\left(t^{\all{X}{\varphi}}\psi\right)^{\varphi [X:=\psi]}$.
 \end{itemize}
 An occurrence of a type variable $X$ is called \emph{bounded}, if it is within the scope of $\Lam{X}{\ldots}$ or $\all{X}{\ldots}$. An occurrence of a type variable which is not bounded is called \emph{free}. The set of free type variables of a term $t$ is denoted by $\FTV(t)$.
The superscript representing the
type is often omitted.

\begin{definition}[Terminal types of $\prts$]\label{def:para isot}
 \begin{enumerate}
     \item $\top\isot.$
     \item $\tau\isot\implies \varphi\to\tau\isot$.
     \item $\tau_1,\tau_2\isot\implies \tau_1\times\tau_2\isot$.
     \item $\tau\isot\implies\all{X}{\tau}\isot$.
 \end{enumerate}
 \end{definition}

\begin{definition}[Stars of $\prts$]\label{def:para star}
 \begin{enumerate}
     \item $\tau\isot\implies \c^{\varphi\to\tau}\equiv\lam{x^\varphi}{\c^\tau}$.
     \item $\tau_1,\tau_2\isot\implies \c^{\tau_1\times\tau_2}\equiv \pairing{\c^{\tau_1}}{\c^{\tau_2}}$.
     \item \label{pi terminal} $\tau\isot\implies\c^{\all{X}{\tau}}\equiv\Lam{X}{\c^\tau}$.
     \end{enumerate}
 \end{definition}

The rewrite rule schemata of the rewriting system $\prts$ are the rewrite rules~$(\beta)$, $(\pi_1)$, $(\pi_2)$, $(\eta)$, $(SP)$ of $\lambda\beta\eta\pi*$, 
\begin{align*}
&(\gentop)         &u^\tau&\to \c^\tau, &\mbox{($u\not\equiv\c^\tau$, $\tau\isot$.)}\\
&(\etatop)\quad & \lam{x^\tau}{t \c^\tau} &\to t, &
\mbox{($x\notin\FV(t)$, $\tau\isot$.)}\\
&(\SPtop1)     &\pairing{\l u}{\c^\tau}&\to u,&\mbox{($u$ has type
  $\varphi\times \tau$, $\tau\isot$.)}\\
  &(\SPtop2)     &\pairing{\c^\tau}{\r u}&\to u, &\mbox{($u$ has type
  $\tau\times \psi$, $\tau\isot$.)}
 \end{align*}
 and the following two:
  \begin{align*}
&(\beta^2)\ \ (\Lam{X}{t})\varphi\to t[X:=\varphi].\qquad\qquad
(\eta^2)\ \ \Lam{X}{s X} \to s,\ \ \mbox{($X\notin\FTV(s)$)}.
\end{align*}
 This completes the definition of $\prts$.

In~\cite{CD}, to show SN of the rewriting system
$\prts$, they tried to prove that
every term of $\prts$ in the $\gentop$-normal form is SN.  But they observed that the set of $\gentop$-normal form is
not closed under $\beta^2$-reduction; $(\Lambda X.\, \lambda x ^ X.\, \lambda y ^{X \to
Y}.\, yx) \top$ is in $\gentop$-normal form,
but its reduct $u\equiv \lambda x ^\top.\, \lambda y^{\top\to Y}.\, yx$ is not,
as $u\to_\gentop \lambda x ^\top.\, \lambda y^{\top\to Y}.\,y *$.

We will  prove SN of $\prts$.

Following \protect{\cite{GTL}}, we consider:

\begin{definition}[Neutral]\label{polymorphism:neutral} A term is  \emph{neutral} if it is not of the form
 $\pairing{u}{v}$, $\lam{x}{v}$, or $\Lam{X}{u}$. \end{definition}

\begin{definition}\label{def:can free}\begin{enumerate}
\item We say a term of $\prts$ is \emph{star-free}, if it has no  subterm  $\c^\tau$ with $\tau\isot$.

\item A set $A$ of terms is called \emph{variant-closed}, provided for any $t\in A$, and  for any variable $z^\top$  not occurring in $t$,
$t[\c^\top:=z^\top]\in A$. 
\end{enumerate}
\end{definition}

As in $\rts$, we consider (CR0) of Lemma~\ref{6.2.3} and variant-closedness  to define a \emph{reducibility candidate}~\cite{GTL}.
	 
 \begin{definition}\label{def:RC}
  A \emph{reducibility candidate}~(RC for short)  of type $\varphi$ is a set $\R$ of
  terms of type $\varphi$ such that:
  \begin{description}
   \item[(CR0)]\ If $\c^\varphi$ is defined, then $\c^\varphi\in\R$. Moreover $\R$ is variant-closed.
    \item[(CR1)]\ If $t^\varphi \in \R$, then $t^\varphi$ is SN.
\item[(CR2)]\ If $t^\varphi \in \R$ and $t^\varphi\to t'$, then $t' \in \R$.
   \item[(CR3)]\ If $t^\varphi$ is neutral, and any reduct of $t^\varphi$ is in $\R$, then $t^\varphi \in \R$.
  \end{description}
    \end{definition}

\begin{lemma}(CR0) and (CR3) implies 
\begin{description}
 \item[(CR4)]\ If $t^\varphi$ is  a variable, then $t$
	    is in $\R$.
\end{description}
\end{lemma}
\begin{proof}The proof is exactly the same as the proof of Lemma~\ref{varred}.\end{proof}
\begin{definition}\label{def:miscrc}
\begin{enumerate}
\item Let $\SN^{\psi}$ be the set of SN terms of type $\psi$.
\item
For  an RC
 $\R_i$ of type $\varphi_i$   of type $\psi_i$ $(i=1,2)$,  then
   \begin{align*}  
 \R_1\times \R_2& = \{ t^{\varphi_1\times\varphi_2} \mid \pi_i t  \in \R_i \ (i=1,2)\}, \\
   \R_1\to\R_2     &= \{ t^{\varphi_1\to\varphi_2} \mid \forall u (u \in \R_1 \implies t u \in \R_2)\}.
    \end{align*}
    \end{enumerate}
    \end{definition}
    
    \begin{lemma}\label{lem:rc}
    \begin{enumerate}
    \item  \label{lem:nonvoid}
For any type $\psi$, $\SN^{\psi}$ is an RC. 

    \item \label{assert:rc prod arrow}If $\R_i$ is an RC of type $\varphi_i$ $(i=1,2)$, then 
    \begin{enumerate}\item \label{prodRC} $\R_1\times\R_2$ is an RC of type $\varphi_1\times\varphi_2$, and
    \item \label{arrowRC} $\R_1\to\R_2$ is an RC of type $\varphi_1\to\varphi_2$.
    \end{enumerate}
    \end{enumerate}
    \end{lemma}
    \begin{proof} \eqref{lem:nonvoid} (CR0):  The variant-closedness is essentially the proof of Lemma~\ref{lem:rel}~\eqref{red rel} for $\varphi$ being atomic.  (CR1): By the definition of $\SN^{\varphi_2}$. (CR2): If $t\in\SN^{\varphi_2}$ and $t\to t'$, then $t'\in\SN^{\varphi_2}$.
 (CR3): Let $t$ be a neutral term of type $\varphi_2$ such that any reduct $t'$ of $t$ is in $\SN^{\varphi_2}$. Then $t$ is in $\SN^{\varphi_2}$. 

\medskip
\eqref{assert:rc prod arrow}
The proof of (CR0), $\ldots$, (CR3) is the proof of Lemma~\ref{6.2.3} for $\varphi$  being a function type or a product type.
But  `by induction hypothesis (CR$k$) on $\varphi_i$'  should be replaced by  `by (CR$k$) of RC $\R_i$.'
The variant-closedness is  the proof of Lemma~\ref{lem:rel}~\eqref{red rel} for corresponding $\varphi$.
But `reducible' should be replaced by ``in $\R$' or `in $\S$.'
\end{proof}

For a type $\varphi$, a sequence $\vec{X}$ of distinct type variables $X_1,\ldots,X_m$, and a sequence $\vec{\psi}$ of types $\psi_1,\ldots,\psi_m$,
let   $\varphi[\vec{X}:=\vec{\psi}]$
  be the simultaneous substitution.
  
  \begin{definition}[Parametric Reducibility]  \label{def:pred} Suppose that
  \begin{enumerate}
  \item  $\varphi$ is a type;
  \item  a sequence
  $\vec{X}$ of distinct type variables $X_1,\ldots,X_m$ contains all   free
type variables of $\varphi$;
\item  $\vec{\psi}$ is a sequence of
  types $\psi_1,\ldots,\psi_m$; and
  \item
 $\vec{\R}$ is a
sequence of RCs $\R_1,\ldots,\R_m$ of corresponding
  types $\vec{\psi}$.
  \end{enumerate} 
  Define a set
 $\RED_\varphi [\vec{X}:=\vec{\R}]$ of terms of type
  $\varphi[\vec{X}:=\vec{\psi}]$  as follows:
 \begin{enumerate}
  \item If $\varphi=\top$, $\RED_\varphi
	  [\vec{X}:=\vec{\R}]=\SN^\top$;
   \item  If $\varphi=X_i$,   $\RED_\varphi
	  [\vec{X}:=\vec{\R}] = \R_i$;
   \item   If $\varphi\equiv \varphi'\circ\varphi''$,  
   $\RED_\varphi
	 [\vec{X}:=\vec{\R}] = \RED_{\varphi'}
	 [\vec{X}:=\vec{\R}]\;\circ\;
	 \RED_{\varphi''}[\vec{X}:=\vec{\R}]$ $(\circ=\to, \times)$ where the latter $\to,\times$ are defined in Definition~\ref{def:miscrc};

   \item \label{pred:Pi} If  $\varphi\equiv \all{Y}{\varphi' }$, $Y$ is not free in $\vec \psi$ and $Y\ne X_i$ ($i=1,\ldots,m$), then $\RED_{\varphi}
	 [\vec{X}:=\vec{\R}]$ is the set of terms
	$ t^{\all{Y}{\varphi'  }[\vec{X}:=\vec\psi]}$ such that for any type $\psi$ and  any RC $\S$ of type $\psi$, 
$( t \psi )^{\varphi' [\vec X, Y := \vec \psi, \psi]} \in \RED_{\varphi'  }[\vec{X},Y:=\vec{\R}, \S ] $. 
\end{enumerate}
\end{definition}

  \begin{lemma} \label{lem:indeed}
  Under the conditions of Definition~\ref{def:pred},
$\RED_\varphi [\vec{X}:=\vec{\R}]$ is an RC
 of type $\varphi[\vec{X}:=\vec{\psi}]$.
  \end{lemma}	 
\begin{proof} By induction on $\varphi$. First consider the case 
 $\varphi\equiv\all{Y}{\varphi'}$. 
 Without loss of generality, $Y$ does not occur free in $\vec\psi$.
 Let $\S$ be an RC of type $\varphi''$.  By induction hypothesis,
	       \begin{align} \T:=\vec\RED_{\varphi'}[\vec{X},Y:=\vec{\R},\S], \mbox{ is an RC.} \label{ind}\end{align}
 \begin{description}
 \item[(CR0)]\ 
 
\ Let $\c^{\all{Y}{\varphi'}[\vec{X}:=\vec{\psi}]} \varphi''\to s$ where $Y$ is not free in $\vec{\psi}$. We will verify $s\in \T$.
Then 
$s\equiv \c^{\varphi'[\vec X,Y:=\vec\psi,\varphi'']}$.  By \eqref{ind},
$s \in \T$.

Thus $\c^{\all{Y}{\varphi'}[\vec{X}:=\vec{\psi}]} \varphi''\in \T$ by
	    (CR3).
	    So
$\c^{\all{Y}{\varphi'}[\vec{X}:=\vec{\psi}]} \in \RED_{\all{Y}{\varphi'}}[\vec{X}:=\vec\R]$.

  \item[(CR1)]\ Let $t \in \RED_\varphi [\vec{X}:=\vec{\R}]$. Then
 $t \varphi'' \in \T$  by  Definition~\ref{def:pred}. 
By \eqref{ind} and (CR1) of $\T$,  $t \varphi''$ is SN. So  $t$ is SN.
	       
\item[(CR2)]\ Let $t \in \RED_\varphi [\vec{X}:=\vec{\R}]$. 
 Then
 $t \varphi'' \in \T$ by  Definition~\ref{def:pred}. 
	 Assume $t\to t'$. Then
      $t \varphi''\to t' \varphi''$.   By \eqref{ind} and (CR2) of $\T$, $t' \varphi'' \in \T$. So $t' \in \RED_\varphi [\vec{X}:=\vec{\R}]$.

  \item[(CR3)]\ Suppose that $t$ is neutral and that $t'\in\RED_\varphi [\vec{X}:=\vec{\R}]$ whenever  $t\to t'$. Let  $t\varphi''\stackrel{\Delta}\to s$.  
  As $t$ is neutral, we have two cases:
  \begin{enumerate}
  \item $\Delta\equiv t\varphi''$:
  Then $s$ is $\c^{\varphi' [\vec X, Y:=\vec \psi, \varphi'']}$, because $t$ is neutral. By \eqref{ind} and (CR0) of $\T$, $s \in \T$.
\item Otherwise,  $s\equiv t' \varphi''$ with $t \stackrel{\Delta}\to t'$. $s\in
	       \T$
	       by $t'\in\RED_\varphi
	       [\vec{X}:=\vec{\R}]$.
	       \end{enumerate}	       By \eqref{ind} and (CR3) of $\T$, $t \varphi'' \in \T$. So $t \in \RED_\varphi [\vec{X}:=\vec{\R}]$.
\end{description}

 To prove the variant-closedness of $\RED_{\all{Y}{\varphi'}} [\vec{X}:=\vec{\R}]$, take an arbitrary $t$ from the set and a variable $z^\top$ not occurring in $t$.
 Then $t\varphi'' \in \RED_{\varphi'} [\vec{X},Y:=\vec{\psi},\S]$ for every RC $\S$ of type $\varphi''$.
 $z^\top$ does not occur in $t\varphi''$.
 By induction on $\varphi'$, $\RED_{\varphi'} [\vec{X},Y:=\vec{\psi},\S]$ is variant-closed. 
 So, $ (t \varphi'' )[\c^\top := z^\top ] \equiv t [\c^\top := z^\top ] \varphi''$ is in $\RED_{\varphi'} [\vec{X},Y:=\vec{\psi},\S]$.
 Thus $ t \in \RED_{\all{Y}{\varphi'}} [\vec{X}:=\vec{\R}]$.
 To sum up, $\RED_{\all{Y}{\varphi'}} [\vec{X}:=\vec{\R}]$ is variant-closed.
 
 \medskip
 When $\varphi\not\equiv\all{Y}{\varphi'}$, we can prove (CR0), (CR1), (CR2), and (CR3) of $\RED_\varphi
	       [\vec{X}:=\vec{\R}]$,
  by induction hypotheses on $\varphi$ and Lemma~\ref{lem:rc}.
\end{proof}

\begin{lemma}\label{sub} Suppose  that 
\begin{enumerate}
\item $\varphi,\psi$ are types, $Y$ is a type variable;
\item
   a sequence $\vec{X}$  of distinct type variables $X_1,\ldots,X_m$ contains all  free
type variables of $\varphi[Y:=\psi]$ and those of $\psi$;
\item
   $X_i\ne Y$ $(i=1,\ldots,m)$; and 
\item
 $\vec\R$ is a sequence of RCs $\R_1,\ldots,\R_m$. 
 \end{enumerate}
 Then
$$\RED_{\varphi[Y:=\psi]}[\vec{X}:=\vec{\R}] = \RED_{\varphi}
 [\vec{X},Y:=\vec{\R}, \RED_\psi
 [\vec{X}:=\vec{\R}] ].$$
\end{lemma}
\begin{proof}
 By induction on $\varphi$.
\end{proof}

\begin{lemma}\label{lem:weakened sufficient condition poly}Let $\R,\S$ be RCs of type $\varphi,\psi$. 
   \begin{enumerate}
   \item  \label{pairing poly} Let $u^\varphi,v^\psi$ be any terms.
   $\pairing{u^\varphi}{v^\psi}\in \R\times\S$, provided that
	 \begin{enumerate}
	  \item 
		$u\in \R$ and $v\in \S$;

	  \item if
		$u\equiv \l w$ and $v \equiv \c^\psi$,
		then $w\in \R\times\S$; and 
	  \item  if
		$v\equiv \r w$ and $u \equiv \c^\varphi$,  then $w\in \R\times\S$.
\end{enumerate}

   \item  \label{lambda poly} Let $v^\psi$ be any term.
   $\lam{x^\varphi}{v^\psi}\in \R\to\S$, provided that
	 \begin{enumerate}
	  \item 
 $v^\psi[x^\varphi:=u^\varphi]\in\S$ for every  \emph{possibly non-star-free term} $u^\varphi\in \R$; and

	  \item if $v\equiv w^{\varphi\to\psi} \c^\varphi$
		and $x^\varphi\notin \FV(w^{\varphi\to\psi})$,
		then $w^{\varphi\to\psi} \in \R\to\S$.
		      \end{enumerate}
  \end{enumerate} \end{lemma}
\begin{proof} The proof is similar to the proof of Lemma~\ref{lem:weakened sufficient condition T}.\end{proof}

\begin{lemma}[Universal abstraction]\label{lem:uabs}
Suppose that 
\begin{enumerate}
\item 
$\varphi$ is a type;

\item 
a sequence $\vec{X}$ of distinct type variables $X_1,\ldots,X_m$ contains all free type variables of $\all{Y}{\varphi}$;

\item  
$X_i\ne Y$ $(i=1,\ldots, m)$,  $\vec{\psi}$ is a sequence of types $\psi_1,\ldots,\psi_m$;

\item
$\vec{\R}$ is a sequence of RCs $\R_1,\ldots,\R_m$ of types $\vec\psi$; 
\item
 $Y$ does not occur free in $\vec\psi$; 

\item $w^{\varphi[\vec X :=\vec\psi]}$ is a term; and

\item \label{ass} for any type $\psi$ and
 any RC $\S$ of type  $\psi$, 
 \begin{align*}(w[Y:=\psi ])^{ \varphi [\vec X, Y:=\vec\psi,\psi]} \in\RED_{\varphi}[\vec{X},Y:=\vec{\R},\S ].\end{align*}
\end{enumerate}
Then
$\Lam{Y}{w} \in \RED_{\all{Y}{\varphi}}[\vec{X}:=\vec{\R}]$.\end{lemma}

\begin{proof}$\SN^Y$ is an RC, by Lemma~\ref{lem:rc}~\eqref{lem:nonvoid}. By  assumption~\eqref{ass}, 
\begin{align} w \in \RED_\varphi [\vec X, Y := \vec \R, \SN^Y]. \label{ddag}\end{align}
By
  (CR1) of this RC, $w$ is SN.  By Definition~\ref{def:pred}~\eqref{pred:Pi}, we have only to verify: 
  \begin{align}(\Lam{Y}{w}) \psi
  \in\RED_{\varphi}[\vec{X},Y:=\vec{\R},\S ], \mbox{for every
  type $\psi$ and RC $\S$ of type $\psi$}. \label{uuu}
  \end{align}
  The proof is by $\wfi{\left(\left\{ w^{\varphi[\vec X:=\vec\psi]} \mid\ \mbox{\eqref{ddag} holds}\right\},\ \to\right)}$ where $\to$ is the rewrite relation.
  Let $(\Lam{Y}{w}) \psi\stackrel{\Delta}{\to} s$. 
We have five cases. We verify $s\in \T:= \RED_\varphi[\vec{X},Y:=\vec{\R},\S ]$.
 \begin{enumerate}
 \item $\Delta\equiv (\Lam{Y}{w}) \psi$ is a redex of $(\gentop)$: Then 
 $s\equiv \c^{\varphi [\vec X, Y := \vec{\psi}, \psi]}$. By (CR0) of $\T$.

 \item $\Delta\equiv (\Lam{Y}{w}) \psi$ is a redex of $(\beta^2)$:  Then  $s\equiv w[Y:=\psi ]$. By assumption~\eqref{ass}.

 \item $\Delta\equiv (\Lam{Y}{w})$ is a redex of $(\gentop)$: Then $s\equiv \c^{\all{Y}{\varphi[\vec X:=\vec \psi]}}\psi$. 
By (CR0), $\c^{\all{Y}{\varphi[\vec X:=\vec \psi]}}\in \RED_{\all{Y}{\varphi}}[\vec X:=\vec\R]$. Hence  $s\in \T$ by Definition~\ref{def:pred}. 

 \item $\Delta\equiv (\Lam{Y}{w})$ is a redex of $(\eta^2)$: Then this case coincides with the second case.
 
\item Otherwise, for some $w'$, $s\equiv (\Lam{Y}{w'}) \psi$ and $w \to w'$. By the WF induction
hypothesis.
\end{enumerate}
Thus $s\in \T$.
 So the statement~\eqref{uuu} follows from (CR3) of $\T$. \end{proof}

\begin{lemma}[Universal application]\label{lem:uapp} Suppose that
\begin{enumerate}
 \item 
$\varphi,\psi$ are types, $Y$ is a type variable;
\item
a sequence  $\vec{X}$ of distinct type variables $X_1,\ldots,X_m$  contains all  free
type variables of $\varphi[Y:=\psi]$ and those of $\psi$;
\item
$X_i\ne Y$ $(i=1,\ldots,m)$,  
$\vec{\psi}$ is a sequence of types $\psi_1,\ldots,\psi_m$; and
\item
  $\vec\R$ is a sequence of RCs $\R_1,\ldots,\R_m$ of types $\vec\psi$.
\end{enumerate}
  Then~\footnote{\cite[Lemma~14.2.3]{GTL}  corresponding to this lemma has a typo: ``$t V$'' should be ``$t (V[\underline{U}/\underline{X}])$.''}
\begin{align*}w \in
 \RED_{\all{Y}{\varphi}}[\vec{X}:=\vec{\R}]\implies 
 w \left(\psi[\vec{X}:=\vec{\psi}]\right)\in\RED_{\varphi[Y:=\psi ]}[\vec{X}:=\vec{\R}].
 \end{align*}\end{lemma}
 
\begin{proof} By Lemma~\ref{lem:indeed}, $\RED_{\psi} [\vec{X}:=\vec{\R}]$ is an RC of type $\psi[\vec X := \vec \psi]$.
By the premise and  Definition~\ref{def:pred}~\eqref{pred:Pi} with $\psi:=\psi[\vec{X}:=\vec{\psi}]$, 
\begin{align*}w  \left(\psi[\vec{X}:=\vec{\psi}]\right) 
 \in\RED_\varphi   [\vec{X},Y:=\vec{\R},\RED_{\psi} [\vec{X}:=\vec{\R}]\, ].
 \end{align*} So Lemma~\ref{sub} implies the conclusion.\end{proof}
\begin{lemma} \label{lem:q}
 Let $\tau\isot$. Then
\begin{enumerate}
    \item \label{assert:q:1} $\tau$ is not of the form $\cdots\to\cdots\to \varphi$ where $\varphi\not\isot$.
    \item \label{assert:q:2} $\c^\tau$ is defined and $\FV(\c^\tau)=\emptyset.$
\end{enumerate}
  \end{lemma}
\begin{proof} By induction on $\tau$.\end{proof}

The following is the counterpart of Lemma~\ref{lem:rel}:
\begin{lemma}\label{lem:closure}
In $\prts$, for every star-free term $t$, 
\begin{enumerate}
\item \label{assert:type subst} a term $t[X_1,\ldots,X_m :=\psi_1,\ldots,\psi_m]$ is star-free  for all distinct type variables $X_1,\ldots,X_m$ and for all types $\psi_1,\ldots,\psi_m$; and 
\item \label{assert:term subst} a term $t[x_1^{\varphi_1},\ldots,x_n^{\varphi_n}:=u_1^{\varphi_1},\ldots,u_n^{\varphi_n}]$ is star-free  for all distinct  variables $x_1^{\varphi_1},\ldots,x_n^{\varphi_n}$ and for all star-free terms $u_1^{\varphi_1},\ldots,u_n^{\varphi_n}$. 
\end{enumerate} 
\end{lemma}
\begin{proof}    By induction on $t$. Let 
\begin{align*}
&\Theta=[X_1,\ldots,X_m :=\psi_1,\ldots,\psi_m]\\
&\theta=[x_1^{\varphi_1},\ldots,x_n^{\varphi_n}:=u_1^{\varphi_1},\ldots,u_n^{\varphi_n}].
\end{align*}
  The proof proceeds by cases according to the form of $t$. By the definition of  $\prts$, $t$ is not a term constant, because otherwise $t$ is $\c^\top$. 
 \begin{itemize}

\item $t$ is a  variable: Then \eqref{assert:type subst} is by Lemma~\ref{lem:q}~\eqref{assert:q:2}. \eqref{assert:term subst} is clear.

\item $t$ is
an abstraction, or an application: By induction hypotheses.

\item $t\equiv \Lam{Y}{w}$ such that $X_i\not\equiv Y$ and $Y$ does not occur free in any $\psi_i$:
 By induction hypothesis, $w\Theta$ and $w\theta$ are star-free. 
 So,  none of $t\Theta\equiv \Lam{Y}{w\Theta}$ and $t\theta\equiv \Lam{Y}{w\theta}$ is a star term. 

\item $t\equiv w\psi$: Then, by induction hypothesis, $w  \Theta $ and $w\theta$ are star-free. Hence, none of $t \Theta  \equiv w  \Theta (\psi  \Theta )$ and $t \theta \equiv (w\theta) \psi$ is a star-term.
\end{itemize}
This completes the proof of Lemma~\ref{lem:closure}.
\end{proof} 
\begin{lemma}\label{key statement for arrow for poly}Suppose that
\begin{enumerate}
\item \label{premise:star-free} $v^\psi$ is a star-free;
\item a sequence $\vec X$  of distinct type variables $X_1,\ldots,X_m$ contains all free type variables of  $v^\psi$; and
\item  $\vec\R$ is a sequence of  RCs $\R_1,\ldots,\R_m$  of
	  types $\vec\psi\equiv\psi_1,\ldots,\psi_m$; 
\end{enumerate}
If $v[x^\varphi:=u^\varphi]\in \RED_\psi [\vec{X}:=\vec{\R}] $ for every  star-free $u^\varphi\in \RED_\varphi [\vec{X}:=\vec{\R}]$, 
then $\lam{x^\varphi}{v^\psi}\in \RED_{\varphi\to\psi} [\vec{X}:=\vec{\R}]$ and star-free.\end{lemma}
\begin{proof}
 Let $w^\varphi\in  \RED_\varphi [\vec{X}:=\vec{\R}]$. Because $\RED_\varphi [\vec{X}:=\vec{\R}]$ is variant-closed, there is a star-free  term $u\in   \RED_\varphi [\vec{X}:=\vec{\R}]$ such that $u\stackrel{*}{\to} w$.  
We have $v[x:=u]\stackrel{*}{\to} v[x:=w]$.
So, by the premise $v[x:=u]\in \RED_\psi [\vec{X}:=\vec{\R}] $ and (CR2), we have $v[x:=w]\in \RED_\psi [\vec{X}:=\vec{\R}] $. 
 Because of the premise~\eqref{premise:star-free} , Lemma~\ref{lem:weakened sufficient condition poly}~\eqref{lambda poly} implies $\lam{x^\varphi}{v^\psi}\in  \RED_{\varphi\to\psi} [\vec{X}:=\vec{\R}]$, while $\lam{x}{v}$ is star-free.
\end{proof}

 \begin{theorem}[Relativized Reducibility]\label{prop:poly}
 Suppose that
\begin{enumerate}
\item	 \label{premise:star free poly} $t^\varphi$ is a \emph{star-free} term;

\item a  sequence of distinct variables $x_1^{\varphi_1},\ldots, x_n^{\varphi_n}$  contains all  free variables of $t^\varphi$;

\item a sequence $\vec X$  of distinct type variables $X_1,\ldots,X_m$ contains all free type variables of  $t$;

\item  $\vec\R$ is a sequence of  RCs $\R_1,\ldots,\R_m$  of
	  types $\vec\psi\equiv\psi_1,\ldots,\psi_m$; 
	  
\item  \label{premise:afo} $u_i^{\varphi_i[\vec{X}:=\vec{\psi}]}$ is in  
$\RED_{\varphi_i} [\vec{X}:=\vec{\R}]$ and is  \emph{star-free} ($i=1,\ldots,n$); and

\item $t[\vec{X}:=\vec{\psi}][\vec{x}:=\vec{u}]$ is the term
  obtained from $t[\vec X:=\vec\psi]$ by simultaneously substitution of $u_1^{\varphi_1[\vec X := \vec \psi]},\ldots,u_n^{\varphi_n[\vec X := \vec \psi]}$ into $x_1^{\varphi_1[\vec X := \vec \psi]},\ldots, x_n^{\varphi_n[\vec X := \vec \psi]}$. 
 \end{enumerate} 
  Then 
$t[\vec{X}:=\vec{\psi}][\vec{x}:=\vec{u}]$  is in $\RED_\varphi
 [\vec{X}:=\vec{\R}]$.
\end{theorem}
\begin{proof} By the premise~\eqref{premise:star free poly}, the premise~\eqref{premise:afo} and Lemma~\ref{lem:closure}, \begin{align}t[\vec{X}:=\vec{\psi}][\vec{x}:=\vec{u}]\ \mbox{is star-free.}\label{uuvv}\end{align}

By induction on $t$.  The proof proceeds by cases according to the form of $t$. By the premise~\eqref{premise:star free poly}, $t$ is not a star term. Then we have five cases.
\begin{enumerate}
\item $t$ is a variable $x_i$: Then $t[\vec X:=\vec\psi][\vec x:=\vec u]\equiv u_i^{\varphi_i[\vec X:=\vec \psi]}$ is in $\R_i$ by the premise~\eqref{premise:afo}.

\item $t$ is a pairing: We can prove this case, similarly as in the proof of Theorem~\ref{prop}, by using \eqref{uuvv} and Lemma~\ref{lem:weakened sufficient condition poly}~\eqref{pairing poly}.

\item $t$ is a $\lambda$-abstraction: We can prove this case, similarly as in the proof of Theorem~\ref{prop}, by using \eqref{uuvv} and Lemma~\ref{key statement for arrow for poly}.

\item   $t\equiv w_1^{\sigma_1}w_2^{\sigma_2}$ where $\sigma_1\equiv \sigma_2\to\varphi$:
If a free type variable occur in $w_i$ $(i=1,2)$, then it does so in $w_1 w_2$. So,
 by induction hypotheses, $w_i [\vec{X}:=\vec{\psi}][\vec{x}:=\vec{u}] \in \RED_{\sigma_i}
 [\vec{X}:=\vec{\R}]$. By $\sigma_1\equiv \sigma_2\to\varphi$, Definition~\ref{def:pred} and Definition~\ref{def:miscrc}, we have $t[\vec{X}:=\vec{\psi}][\vec{x}:=\vec{u}] \in \RED_\varphi
 [\vec{X}:=\vec{\R}]$.

 \item  $t\equiv  (\Lam{Y}{w})^ {\all{Y}{\sigma}}$ where $X_i\neq Y$ and $Y$ does not occur free in any $\varphi_i[\vec X:=\vec\psi]$: Then
 by the induction hypothesis, for any type $\psi$ and any RC $\S$ of $\psi$, 
$ w [\vec{X},Y:=\vec{\psi},\psi][\vec{x}:=\vec{u}]$ is in $\RED_\sigma
 [\vec{X},Y:=\vec{\R},\S]$. 
 Since $Y$ occurs in no $\vec u$ without loss of generality, we have $w [\vec{X}:=\vec{\psi}][\vec{x}:=\vec{u}][Y:=\psi]\in \RED_\sigma
 [\vec{X},Y:=\vec{\R},\S]$.
 By Lemma~\ref{lem:uabs}, $(\Lam{Y}{w}) [\vec{X}:=\vec{\psi}][\vec{x}:=\vec{u}] $ is in $\RED_{\all{Y}{\sigma}}
 [\vec{X}:=\vec{\R}]$. 
 
 \item   $t\equiv  w^ {\all{Y}{\sigma}} \psi$: Then by the induction hypothesis,
$ w [\vec{X}:=\vec{\psi}][\vec{x}:=\vec{u}]\in\RED_{\all{Y}{\sigma}}
 [\vec{X}:=\vec{\R}]$. By Lemma~\ref{lem:uapp}, $ w [\vec{X}:=\vec{\psi}][\vec{x}:=\vec{u}] \left(\psi [\vec X:=\vec \psi ]\right) 
 \in \RED_{\sigma[Y:=\psi]}
 [\vec{X}:=\vec{\R}]$. This term is just 
 $(w \psi)[\vec{X}:=\vec{\psi}][\vec{x}:=\vec{u}]$.
 \end{enumerate}
This completes the proof of  
Theorem~\ref{prop:poly}. 
\end{proof} 
\begin{definition}\label{def:poly:reducibility} A term $t^\varphi$ is called \emph{reducible}, if
for some sequence 
of distinct type variables $X_1,\ldots,X_m$ containing  the free type variables of a type $\varphi$, 
\begin{align*}
t^\varphi\in\RED_\varphi [X_1,\ldots,X_m:=\SN^{X_1},\ldots,\SN^{X_m}].\end{align*}\end{definition}

\begin{theorem}\label{poly:reducibility theorem}
Any  term $t^\varphi$ is in $\RED_\varphi [X_1,\ldots,X_m:=\SN^{X_1},\ldots,\SN^{X_m}]$.\end{theorem}

\begin{proof}Amy  star-free term is reducible, by
(CR4) and by Theorem~\ref{prop:poly} with $u_i^{\varphi_i}:=x_i^{\varphi_i}$,
 $\psi_j:=X_j$ and $\R_j:=\SN^{X_j}$. 
 Hence a star-free term $t [ \c^\top := z^\top ]$ is reducible for some variable $z^\top$. Because $t [ \c^\top := z^\top ] \rtc t$.  $t$ is reducible, by (CR2). \end{proof}

\begin{corollary}  $\prts$ satisfies SN and CR.
\end{corollary}
\begin{proof}  SN
 follows from (CR1) and Theorem~\ref{poly:reducibility theorem}. $\prts$ is weakly confluent by \cite[Proposition~2.5]{CD}.  
  So, Newman's lemma~\cite{Newman} implies CR of $\prts$. \end{proof}

\subsection{Parametric terminal types}

According to \cite{Has}, in the parametric polymorphism, a type $\all{X}{(X\to X)}$ is a terminal type. 
We will add the following clauses to the definition of $\isomorphicT$ and $\c^\varphi$ of $\prts$.
\begin{align}
\label{sharp} \begin{cases}
\all{X}{(X\to X)}\isot\\
\c^{\all{X}{(X\to X)}}:\equiv \Lam{X}{\lam{x^X}{x^X}}
\end{cases}
\end{align}
Then, for a suitable condition,
\begin{align}\label{obst}
\begin{array}{c c c}
\left(\Lam{X}{t^{X \to X}}\right)\varphi & \to_{\beta^2} & t^{\varphi \to \varphi}\\
\downarrow_\gentop & & \\
\left(\Lam{X}{\lam{x^X}{x^X} }\right)\varphi&\to_{\beta^2}& \lam{x^\varphi}{x^\varphi }.
\end{array}
\end{align}
If $\varphi\not\isot$, then it may not be the case $t^{\varphi\to\varphi} \rtc \lam{x^\varphi}{x^\varphi }.$ 
So, $\prts+\eqref{sharp}$ may be not confluent.

The following rewrite rule schema resolves the confluence problem~\eqref{obst}:
\begin{align*}
&(\gentop_{aux})\quad 
t ^{\varphi \to \varphi}
 \to 
\lam{x^\varphi}{x^\varphi},\\
&\mbox{provided}\\
&\mbox{$t$ is of the form $s^{X \to X} [X:=\varphi]$,} \\
&\mbox{$\varphi\not\isot$,} \\
&\mbox{$X$ does not occur free in the type of any free term variable of $s$, and}\\
&\mbox{$t\not\equiv\lam{x^\varphi}{x^\varphi}$.}
 \end{align*}
The rewrite rule schema $(\gentop_{aux})$ can be regarded as an `instance' of a rewrite rule 
\begin{align*}(\gentop)\quad t^{\all{X}{(X\to X)}} \to \c^{\all{X}{(X\to X)}} \quad\mbox{(the left-hand side is not the right-hand side)}.\end{align*}

If we attempt to prove  the SN of $\prts+\eqref{sharp}+(\gentop_{aux})$ by a relativized reducibility candidate method of Section~\ref{sec:polymorphism}, we require 
\begin{align}
\mbox{If $\R$ is an RC of type $\varphi\to\varphi$, then $\lam{x^\varphi}{x^\varphi}\in\R$.} \label{hosi}
\end{align}
It is because $(\gentop_{aux})$ will cause, at least, the following new cases in the proof of Lemma~\ref{lem:weakened sufficient condition poly}:
\begin{itemize}
\item ``Case $\Delta\equiv\pi_1\pairing{u}{v}$ is a redex of
      $(\gentop_{aux})$, $s\equiv \lam{x^\theta}{x^\theta}$ and
      $\varphi_1=\theta\to\theta$ for some type $\theta$.''
\item ``Case $\Delta\equiv(\lam{x}{v}) u$ is a redex of $(\gentop_{aux})$, $s\equiv \lam{x^\theta}{x^\theta}$ and $\varphi_2=\theta\to\theta$ for some type $\theta$.''    
\end{itemize}

If we add the property~\eqref{hosi} in the definition of RC, then we cannot prove ``If $\R,\S$ are RCs of type $\varphi$, then $\R\to\S$ is an RC of type $\varphi\to\varphi$.''
It is because $\R\subseteq\S$ is not always available.
 
\section*{Acknowledgements.} The author thanks K.~Fujita, H.~Goguen, D.~Kesner, K.~Kikuchi, C.~Kop, F.~Pfenning,
Y.~Toyama, and H.~Yokouchi.  The author owes  C.~Kop for the observation on HORPO and WANDA.  This research is supported by a grant of
Graduate School of Science, Tohoku University, Japan.

\appendix

\section{Type-directed expansions}\label{subsec:expansion}

 For the typed $\lambda$-calculus, let a binary relation
 $\to_{\reta}$ ($\to_{\rSP}$) replace a neutral subterm
 occurrence in a non-elimination context with the
 $\eta$~($SP$)-expansion~\cite{Min79}. Neither $\to_{\reta}$ nor
 $\to_{\rSP}$ is stable under contexts.
We call the relation
 $\to:=\to_{\beta\l\r T {\reta}\rSP}$ \emph{Mints' reduction}, as Mints introduced it in \cite{Min77,Min79}. Mints' reduction
 generates the equational theory $\lambda\beta\eta\pi*$,
 and is SN+CR~(\cite{A93,JG95} to cite a few). In \cite{A93}, the author presented a divide-and-conquer lemma to infer
 SN+CR  property of
a reduction system from that property of its subsystems. 
\begin{lemma}[\protect{\cite{A93}}]\label{lem:divide-and-conquer}
If two binary relations $\to_R$ and $\to_S$ on a set $U\ne\emptyset$ have SN+CR property, then so does $\to_{SR}$, provided that we have
\begin{align*}
\forall u,v\in U\left(u\to_S v \implies u^R\stackrel{+}{\to_S} v^R\right),
\end{align*}
where $u^R$ and $v^R$ are the $\to_R$-normal forms of $u$ and $v$ respectively, and $\stackrel{+}{\to_S}$ is the transitive closure of $\to_S$.
\end{lemma}
By inductive arguments, the author proved that  
\begin{align}
\label{goal} t\to_{\beta\pi_1\pi_2 T} s\implies t^{\reta\rSP}\stackrel{+}{\to_{\beta\pi_1\pi_2T}}s^{\reta\rSP}\end{align} 
where $u^{\reta\rSP}$ is  $\to_{\reta\rSP}$-normal form of $u$. By
Lemma~\ref{lem:divide-and-conquer},  SN+CR of Mints' reduction follows. The SN of Mints' reduction implies
 $\to_{\reta}\;=\;\leftarrow_\eta\setminus\leftarrow_\beta$ and
 $\to_{\rSP}\;=\;\leftarrow_{SP}\setminus\leftarrow_{\pi_1}\setminus\leftarrow_{\pi_2}$.
 
  \v{C}ubri\'c  proved the weak normalization of Mints' reduction in \cite{Cubric} and then
published the proof of SN in his thesis~\cite{CubricPhD}. His SN proof is showing \eqref{goal} by proving the commutativity~\cite[Proposition~3.29]{CubricPhD} of $\stackrel{*}{\to}_{\beta\pi_1\pi_2 T}$ and $\stackrel{*}{\to}_{\reta\rSP}$ and the presevation~\cite[Proposition~3.40]{CubricPhD} of $\beta\pi_1\pi_2T$-normal form by $\to_{\reta\rSP}$.
The authors proved \eqref{goal} by mostly inductive argument.
This was part of his generalization of Friedman's theorem for CCCs~\cite{CubricJPAL}.
According to Phil Scott,  the work of   \v{C}ubri\'c  was motivated by the fact that Mints' expansionary rewrites contained mistakes
and Mints' results were wrong!  He resolved the issue with a  detailed analysis of $\eta$-expansion. The issue
was extending the Friedman Set-interpretation from a free CCC $C$ into the category of sets into the free CCC $C(X)$ with an infinite set of indeterminates $X$
adjoined. It was  the problem of the faithfulness of the embedding $C \to C(X)$  which led to   \v{C}ubri\'c   finding the mistakes in Mints' work.

  \v{C}ubri\'c, Dybjer, and Scott employed \emph{normalization-by-evaluation}~(NBE) techniques 
to prove directly the decision problem for the free CCC, without needing any Church-Rosser, SN, or even rewriting at all.  But
they did a computability argument at the end to in fact show that, from a traditional viewpoint, they are actually constructing
long $\beta\eta$-normal forms. They added appendix in the proof of paper (mostly by   \v{C}ubri\'c) where for each typed lambda calculi generated by (i) a graph, (ii) a category,
(iii) a cartesian category,  he tried to prove that the NBE decision procedure makes sense, and reduces the problem of the higher
order structure roughly down to the decision problem of the underlying theory. For this purpose, he attempted to prove the transitivity rule of the equality is  admissible in a formalized equational theory, by a similar proof technique of cut-elimination theorem of proof theory.

 Although $\to_{\reta \rSP}$ is not
stable under contexts, 
 the finite
 development-like argument based on $\leftarrow_{\eta SP}$ proves CR of $\to$~\cite{DBLP:journals/entcs/KhasidashviliO95} pointed out that.

  	In \cite{DK94}, Di Cosmo and Kesner proved CR+SN of a reduction
system
$\to_\beta\cup\to_{\reta}\cup\to_{\pi_1}\cup\to_{\pi_2}\cup\to_{\rSP}\cup\to_\gentop$ union the $\beta$-like reductions of sum types.
By showing how substitution and the reduction interact with the context-sensitive rules, they
proved the WCR.  They simulated expansions without
expansions, to reduce  SN of the reduction to SN for the
underlying calculus without expansions, provable by the
standard reducibility method.

The rewriting system $\rts$ of Curien and Di Cosmo is stable under contexts~(i.e.,
  $t\to t'\implies \cdots t\cdots \to
       \cdots t'\cdots$.)
Mints' reduction 
decides the equational theory $\lambda\beta\eta\pi*$.  Mints' reduction is not stable under contexts.

Mints' reduction  fits to  semantic
treatments such as NBE~\cite{BDF04}. See \cite{ACD} in the context of type-checking of dependent type theories).
However, because of the complication of Mints' reduction, in his book~\cite{Min92} on  selected papers of proof theory,
 Mints replaced his reduction with the $\beta\eta$-reduction modulo equivalence relation on terms. His purpose is to give a simple proof of difficult theorems of category theory with 
 typed $\lambda$-calculus and proof theory by using the correspondence objects = types = propositions and arrows = terms = proofs. Mac Lane is interested in his ambition~\cite{MacLane82}.  
 
\end{document}